%% file: notepad.tex
\pdfoutput=1

\documentclass[11pt]{article}
\usepackage[T1]{fontenc}
\usepackage[utf8]{inputenc}
\usepackage{lmodern}
\usepackage{xspace}
\usepackage[protrusion=true,expansion=true]{microtype}
\usepackage{fullpage}
\usepackage{amsfonts,amsmath,amssymb,amsthm, pbox}

\usepackage[usenames,dvipsnames,table]{xcolor}
\usepackage[colorlinks,citecolor=blue,bookmarks=true]{hyperref}

\usepackage{multirow}
\usepackage{chngpage} 

\usepackage{dsfont} 
\usepackage{fullpage}

\usepackage[colorinlistoftodos,textsize=scriptsize]{todonotes}

\usepackage{algorithmicx,algpseudocode,  algorithm}

\usepackage[shortlabels]{enumitem}
\setitemize{noitemsep,topsep=0pt,parsep=0pt,partopsep=0pt}
\setenumerate{noitemsep,topsep=0pt,parsep=0pt,partopsep=0pt}

\usepackage{epigraph}

\def\withnotes{1}

\input{locdef}
\input{glodef}

\title{A Unified Maximum Likelihood Approach for Optimal Distribution Property Estimation}
\author{Jayadev Acharya\thanks{Supported by an MIT-Shell Energy Initiative grant, and Cornell University startup grant.}\\Cornell University\\\tt{acharya@cornell.edu}
\and
Hirakendu Das\\Yahoo!\\\tt{hdas@yahoo-inc.com}\\
\and
Alon Orlitsky\\UC San Diego\\\tt{alon@ucsd.edu}
\and
Ananda Theertha Suresh\\ Google Research\\\tt{theertha@google.com}}

\begin{document}

\maketitle
\thispagestyle{empty}

 \input{abstract3}
\input{introduction3}
\input{introduction2}

\input{problem-statements}
\input{definitions-notations.tex}

\input{outline}
\input{property-by-pml-2}

\input{sharp-concentration}



\input{acknowledgements}

\bibliographystyle{alpha}
\bibliography{abr,masterref}
\appendix
\input{unseen}
\input{entropy-l1dist}

\input{proof-approximation}

\end{document}

%% file: locdef.tex
\newcommand{\dP}{p}
\newcommand{\dQ}{q}
\newcommand{\dU}{u}

\newcommand{\logk}{\log \nsmp}

\newcommand{\tdP}{\tilde {\dP}}

\newcommand{\nsmp}{n}

\newcommand{\absz}{k}
\newcommand{\msmp}{m}

\newcommand{\err}{\delta}
\newcommand{\dst}{\varepsilon}

\newcommand{\ab}{\mathcal{X}}

\newcommand{\smb}{x}

\newcommand{\smby}{y}
\newcommand{\smbz}{z}

\newcommand{\Dk}{\Delta_{\absz}}
\newcommand{\Dgk}{\Delta_{\ge\frac1\absz}}

\newcommand{\dTV}[2]{d_{TV}(#1,#2)}

\newcommand{\dPsmb}{\dP({\smb})}

\newcommand{\dPsmby}{\dP({\smby})}

\newcommand{\dPsmbz}{\dP({\smbz})}
\newcommand{\dPsmbZ}{\dP({Z})}

\usepackage[colorinlistoftodos,textsize=scriptsize]{todonotes}
\ifnum\withnotes=1

\else

\fi

\newcommand{\profile}{\varphi}

\newcommand{\mlt}{\mu}
\newcommand{\Mlt}{N}
\newcommand{\Mlts}[1]{\Mlt_{#1}}
\newcommand{\Mltsmb}{\Mlts{\smb}}

\newcommand{\flnpwrss}[2]{#1^{\underline{#2}}}

\newcommand{\prop}{f}
\newcommand{\propp}[1]{\prop(#1)}

\newcommand{\hprop}{\hat{f}}
\newcommand{\fhat}{\hat{f}}
\newcommand{\hpropz}[1]{\hprop(#1)}

\newcommand{\mP}[1]{\dP_{_{#1}}}

\newcommand{\multiset}[1]{{M}(#1)}

\newcommand{\prv}{\varphi}

\newcommand{\emd}[2]{{R}\left(#1,#2\right)}

\newcommand{\rqed}{\hbox{}~\hfill~$\square$\\[-2ex]}
\newcommand{\maxp}[1]{p_{_{#1}}}

%% file: abstract3.tex
\begin{abstract}
The advent of data science has spurred interest in estimating properties of discrete distributions over large alphabets. Fundamental symmetric properties such as support size, support coverage, entropy, and proximity to uniformity, received most attention, with each property estimated using a different technique and often intricate analysis tools. 

Motivated by the principle of maximum likelihood, we prove that for all these properties, a single, simple, plug-in estimator---profile maximum likelihood (PML)~\cite{OrlitskySVZ04}---performs as well as the best specialized techniques.
We also show that the PML approach is \emph{competitive} with respect to any symmetric property estimation, raising the possibility that PML may optimally estimate many other symmetric properties.

\end{abstract}

%% file: introduction3.tex
\section{Introduction}
\label{sec:introduction}
\subsection{Property estimation}
Recent machine-learning and data-science applications have motivated a
new set of questions about inferring from data.  A large
class of these questions concerns estimating properties of the unknown
underlying distribution.

Let $\Delta$
denote the collection of discrete distributions.
A distribution \emph{property} is a mapping $f:\Delta\to\reals$. 
A distribution property is \emph{symmetric} if it remains
  unchanged under relabeling of the domain symbols. For example,
  \emph{support size} 
  \[
  S(p)=|\sets{x:p(x)>0}|,
  \] or \emph{entropy}
  \[
  H(p)=\sum_x p(x)\log\frac1{p(x)}
  \] are all symmetric properties
  which only depend on the set of values of $\dPsmb$'s and not what
  the symbols actually represent. To illustrate, the two distributions
    $p = ( p(a) = 2/3, p(b) = 1/3)$ and $p' = (p'(a) = 1/3, p'(b) = 2/3)$ have the same entropy.
      In the common setting for these
questions, an unknown underlying distribution $p\in\Delta$ generates
$n$ independent samples $X^n\ed X_1,\upto X_n$, and from this
\emph{sample} we would like to estimate a given property $f(p)$.

An age-old universal approach for estimating distribution
properties is \emph{plug-in estimation}.
It uses the samples $X^n$ to find an approximation
$\hat p$ of $p$, and declares $f(\hat p)$ as the
approximation $f(p)$.

Perhaps the simplest approximation for $p$ is the
\emph{sequence maximum likelihood (SML)}.
It assigns to any sample $x^n$ the distribution $p$ that maximizes
$p(x^n)$.
It can be easily shown that SML is exactly the \emph{empirical frequency}
 estimator
that assigns to each symbol the fraction of times it appears
 in the
sample, $\maxp{X^n}(x)= \frac{\Mltsmb}{n}$, where
$\Mltsmb \ed \Mltsmb(X^n)$, the \emph{multiplicity} of symbol $\smb$,
is the number of times it appears in the sequence $X^n$. We will just write $\Mltsmb$, when $X^n$ is clear from the context . 
For example, if $n=11$, and $X^n={a~b~r~a~c~a~d~a~b~r~a}$, $\Mlts{a} =
5, \Mlts{b} =2, \Mlts{c} = 1,\Mlts{d} = 1$, and $ \Mlts{r} = 2$, and
$\maxp{X^{11}}(a)=5/11$, $\maxp{X^{11}}(b)=2/11$,
$\maxp{X^{11}}(c)=1/11$, $\maxp{X^{11}}(d)=1/11$, and
$\maxp{X^{11}}(r)=2/11$.

While the SML plug-in estimator performs well in the limit of many
samples and its convergence rate falls short of the best-known
property estimates.
For example, 
suppose we sample the uniform distribution over $k$ elements $n=k/2$ times. Since at most $n$ distinct symbols will appear, the empirical distribution will have entropy at most $\log n \le \log k -1$ bits. However from Table~\ref{tab:results},  for large $k$, only $n=O(k/\log k)$ samples are required to obtain a 1-bit accurate estimate.

Modern applications where the sample size $n$ could be sub-linear in the domain size $k$, have motivated many results characterizing the sample complexity of estimating various distribution properties (See $e.g.$,~\cite{paninski2003, BatuDKR02, BarKS01,
  Valiant11b, valiant2011power, WuY14a, WuY15, acharya2014estimating, caferov2015optimal, JiaoVHW15, OrlitskySW16, ZouVV16, BuZLV16}). Complementary to property estimation is the distribution property testing, which aims to design (sub-linear) algorithms  to test whether distributions have some specific property (See $e.g.$,~\cite{Batu01, GoldreichR00, BatuFRSW00, Paninski08, ChanDVV13, canonne2016testing, adk15, diakonikolas2016new}, and~\cite{Canonne15} for a survey). A particular line of work is competitive distribution estimation and testing~\cite{AcharyaDJOP11, AcharyaDJOPS12, AcharyaJOS13, AcharyaJOS13b,ValiantV13, OrlitskyS15}, where the objective is to design algorithms independent of the domain size, with complexity close to the \emph{best possible} algorithm. Some of our techniques are motivated by those in competitive testing.
%
\subsection{Prior results}
Since SML is suboptimal, several recent papers have used diverse and
sophisticated techniques to estimate important symmetric distribution
properties.
\begin{description}
\item[Support size]

$S(p) =|\sets{x:p(x)>0}|$, plays an important role in population and
  vocabulary estimation.  However estimating $S(p)$ is hard
    with \emph{any finite} number of samples due to symbols with
    \emph{negligible} positive probability that will not appear in our
    sample, but still contribute to $S(p)$. To circumvent
  this,~\cite{RaskhodnikovaRSS09} considered distributions in $\Delta$
  with non-zero probabilities at least $\frac1\absz$,
  \[
  \Dgk
  \ed\left\{\dP\in\Delta:
  \dPsmb\in\{0\}\cup\left[\frac1k,1\right]\right\}.\]
For $\Dgk$, SML requires $\absz\log\left(\frac1\dst\right)$ to estimate the support size to an additive accuracy of  $\dst\absz$. Over a series of
work~\cite{RaskhodnikovaRSS09, Valiant11b, WuY15}, it was shown that
the optimal sample complexity of support estimation is
$\Theta\left(\frac{k}{ \log k} \cdot \log^2\frac{1}{\dst}\right)$.

\item[Support coverage] $S_m(p) = \sum_x (1-(1-p(x))^m)$, the
  expected number of elements observed when the distribution is
  sampled $m$ times, arises in many ecological and biological
  studies~\cite{Colwell12}.  The goal is to estimate $S_m(p)$ to an
  additive $\pm \dst m$ upon observing as few samples as possible.
  Good and Toulmin~\cite{good1956number} proposed an estimator that
  for any constant $\dst$, requires $m/2$ samples to estimate
  $S_m(p)$.
Recently,~\cite{OrlitskySW16, ZouVV16} showed that it is possible to estimate $S_m(p)$ after observing only $\cO(\frac{m}{\log m})$ samples. In particular,~\cite{OrlitskySW16} showed that it is possible to estimate $S_m(p)$ after observing only $\cO(\frac{m}{\log m}\cdot {\log\frac1\dst})$ samples. Moreover, this dependence on $m$ and $\dst$ is optimal.
                                                                                                                                                                                                           

\item[Entropy]
$ H(p) =\sum_x p(x)\log\frac1{p(x)}$, the Shannon entropy of $p$ is a central object in information theory~\cite{CoverT06}, and also arises in many fields such as machine learning~\cite{Nowozin12}, neuroscience~\cite{BerryWM97, NemenmanBRS04}, and others. Entropy estimation has been studied for over half a century, and a number of different estimators have been proposed over
  time.  Estimating $H(p)$ is hard with any finite number of samples due to the possibility of infinite support. To circumvent this,
    similar to previous works we consider distributions in $\Delta$ with support size at most  $k$,
   \[\Dk \ed \{ p \in \Delta : S(p) \leq k\}.\]
 The goal is to estimate the entropy of a distribution in $\Dk$ to an additive $\pm \dst$, where $\Dk$ is all discrete distributions over at most $k$ symbols.
  In a recent set of papers~\cite{Valiant11b, WuY14a,
JiaoVHW15}, the min-max sample complexity of estimating entropy to $\pm\dst$
was shown to be $\Theta\left(\frac{\absz}{\log \absz}\cdot \frac1{\dst}\right)$.

\item[Distance to uniform] $\lVert p-u\rVert_1 = \sum_x |p(x)-1/k|$, where $u$
  is a uniform distribution over a known set $\cX$, with $|\cX|=k$. Let $\Delta_\cX$ be the set of distributions
  over the set $\cX$.
For an unknown $\dP\in\Delta_\cX$, to estimate $||{\dP}-{\dU}||_1$ to an additive
$\pm \dst$,~\cite{valiant2011power} showed that $\cO\left(\frac{\absz}{\log \absz}\cdot \frac1{\dst^2}\right)$ samples are sufficient. The dependence was later shown to be tight in~\cite{JiaoHW16}.
\end{description}
\cite{Valiant11b} also proposed a plug-in approach for estimating symmetric properties. We discuss and compare the approaches in Section~\ref{sec:pml-intro}.

\subsection{New results}
Each of the above properties was studied in one or more papers and
approximated by different sophisticated estimators, often drawing from
involved techniques from  fields such as approximation
theory.  By contrast, we show that a single simple plug-in estimator
achieves the state of the art performance for all these problems.

As seen in the introduction for entropy, SML is suboptimal in the large alphabet regime, since it over-fits the estimate on only the \emph{observed symbols} (See~\cite{jiao2014maximum} for detailed performance of SML estimators of entropy, and other properties).
However, symmetric properties of distributions do not depend on the labels of the symbols. For all these properties, it makes sense to look at a sufficient statistic, the data's \emph{profile} (Definition~\ref{def:profile}) that represents the number of elements appearing any given number of times. Again following the \emph{principle of maximum likelihood},~\cite{OrlitskySVZ04, OSVZ11} suggested discarding the symbol labels, and finding a distribution that maximizes the probability of the observed profile, which we call as \emph{profile maximum likelihood (PML)}.

We show that replacing the SML plug-in estimator by PML yields a
unified estimator that is provably at least as good as the best
specialized techniques developed for all of the above properties.
\begin{theorem}[Informal]
There is a unified approach based on PML distribution that achieves
the optimal sample complexity for all the four problems mentioned
above (entropy, support, support coverage, and distance to uniformity).
\end{theorem}
We prove in Corollary~\ref{cor:competitive} that the PML approach is \emph{competitive} with respect to \emph{any symmetric property}.

For symmetric properties, these results are perhaps a justification of Fisher's thoughts on Maximum Likelihood:
\setlength\epigraphwidth{6.1in}
\epigraph{``\textit{Of course nobody has been able to prove that
    maximum likelihood estimates are best under all
    circumstances. Maximum likelihood estimates computed with all the
    information available may turn out to be inconsistent. Throwing
    away a substantial part of the information may render them
    consistent.}''}{R. A. Fisher's thoughts on Maximum Likelihood.}

However, several heuristics for estimating PML has been studied
including approaches motivated by algebraic
approaches~\cite{AcharyaDMOP10}, EM-MCMC algorithms~\cite{orlitsky2004algorithms},~\cite[Chapter 6]{PanT12}, Bethe approximation~\cite{Vontobel12, Vontobel14}. 
As discussed in Section~\ref{sec:pml-intro}, PML estimation reduces to maximizing a monomial-symmetric polynomial over the simplex. We also provide another justification of the PML approach by proving that even approximating a PML can result in sample-optimal estimators for the problems we consider.  We hope that these strong sample complexity guarantees will motivate algorithm designers to design efficient algorithms for approximating PML.
Table~\ref{tab:results} summarizes the results in terms of the
sample complexity.



\begin{table}[htb]
\medskip
\begin{center}
\begin{tabular}{|c|c|c|c|c|c|c|}
\hline
Property & Notation & $\cP$ & SML & Optimal & References
& PML\\ \hline
Entropy & $H(p)$ & $\Delta_k$ & $\frac{\absz}{\dst}$ &
$\frac{\absz}{\log \absz}\frac1{\dst}$ &
\!\!\cite{Valiant11b, WuY14a, JiaoVHW15}\!\! & optimal\footnotemark  \\ \hline
Support size & $\frac{S(p)}{k}$ & \!$\Dgk$\! & $\absz\log\frac1\dst$ &
$\frac{k}{\log k} \log^2 \frac1\dst$ &
\cite{WuY15} & optimal \\ \hline
Support coverage & $\frac{S_m(p)}{m}$ & $\Delta$ & $m$ & $\frac{m}{\log m}\log\frac1\dst$ &\cite{OrlitskySW16} & optimal \\ \hline
Distance to $u$ & $\lVert p-u\rVert_1$ & $\Delta_\cX$ & $\frac{k}{\dst^2}$ &
$\frac{k}{\log k} \frac{1}{\dst^2}$ &
\cite{valiant2011power, JiaoHW16} & optimal \\ \hline
\end{tabular}
\end{center}
\label{tab:results}
\caption{Estimation complexity for various properties, up to a
constant factor. For all properties shown, PML achieves
the best known results. 
Citations are for specialized techniques, PML results are shown in
this paper. Support and support coverage results have 
been normalized for consistency with existing literature.}
\end{table}
\footnotetext{{We call an algorithm optimal if it is optimal up to universal constant factors.}}

\noindent To prove these PML guarantees, we establish two results that are
of interest on their own right.
\begin{itemize}
\item
With $n$ samples, PML estimates any symmetric property of $p$
with essentially the same accuracy, and at most
$e^{3\sqrt n}$ times the error, of any other estimator. 
\item
For a large class of symmetric properties, including all those mentioned above, if there is an estimator that uses $n$ samples, and has an error probability $1/3$, we design an estimator using $O(n)$ samples, whose error probability is nearly exponential in $n$.  We remark that this decay is much faster than applying the median trick.
\end{itemize}
Combined, these results prove that PML plug-in estimators are sample-optimal.

We also introduce the notion of \emph{$\beta$-approximate ML}
distributions, described in Definition~\ref{def:appr}. These
distributions are more relaxed version of PML, hence may be
more easily computed, yet they provide
essentially the same performance guarantees.


\ignore{
Furthermore, note that 
the sorted total variation distance between two distributions $\dP$ and $\dQ$ 
is the least total variation distance between all possible pairs of distributions obtained by relabeling of the symbols of $\dP$, 
with the distribution $\dQ$. 
It can be shown that the sorted total variation distance is at most
$\emd{\dP}{\dQ}$, and thus estimating in EMD is at least as strong as
estimating in sorted total variation distance. 
}
The rest of the paper is organized as follows. In Section~\ref{sec:results}, we formally state our results. In Section~\ref{sec:pml-intro}, we define profiles, and PML. 
In Section~\ref{sec:outline}, we outline the our approach. In
Section~\ref{sec:pml_maximum}, we demonstrate auxiliary results for
maximum likelihood estimators. In Section~\ref{sec:approach}, we
outline how we apply maximum likelihood to support, entropy, and
uniformity, and support coverage.

%% file: introduction2.tex
\ 

\ignore{
For a number of these problems, it is easy to see that the standard approach of plugging in the empirical distribution (we denote it by SML standing for Standard Maximum Likelihood), has sub-optimal sample complexity. For example, suppose we want to estimate the entropy of the uniform distribution $\dU$ over $\absz$ elements. Then, it is easy to see that until we sample $\Omega(k)$ times, a constant fraction of symbols will not appear in the samples, and we will be off by at least a constant from $\log \absz$, the true entropy. Indeed, there have been a number of papers showing the inefficiency of SML for these problems~\cite{Jiao, ??, ??}.  
}
\ignore{
\subsection{Prior Work}
ML estimators are appealing because they are a unified and simple
technique for estimating a property $f(p)$. The ML estimator outputs
$f(\maxp{X^n})$. However, the convergence rate of the above estimator
is slow.  In particular, it is well known that for estimating several
fundamental properties such as entropy, support size, unseen
estimation, the number of samples should be $\Omega(k)$. To overcome
this, estimators for each of these properties have been studied
extensively in the past decade: entropy~\cite{paninski2003,
Valiant11b, WuY14a, JiaoVHW15}, support size~\cite{RaskhodnikovaRSS09,
Valiant11b, WuY15}, unseen species estimation~\cite{ZouVV16, OrlitskySW16}, and
distance to uniformity~\cite{Valiant11a}. The sample complexity of
techniques (including ML) can be seen Table~\ref{tab:results}. We
expand on the contributions below.
\noindent\textbf{Entropy} estimation has been studied for more than half a century, and
a number of different estimators have been proposed over time
(see~\cite{paninski2003} and references therein).~\cite{paninski2003}
showed the existence of an estimator with sub-linear sample complexity
in $\absz$. In a recent set of papers~\cite{Valiant11b, WuY14a,
JiaoVHW15}, the exact min-max sample complexity of estimating entropy
was shown to be $\Theta\left(\frac{\absz}{\log \absz}\frac1{\dst}\right)$.
\noindent\textbf{Support} estimation via SML has a sample complexity of
$\absz\log\left(\frac1\dst\right)$. Over a series of
work~\cite{RaskhodnikovaRSS09, Valiant11b, WuY15}, it was shown that
the optimal sample complexity of support estimation is
$\Theta\left(\frac{k}{ \log k} \cdot \log^2 \left(\frac{1}{\dst}\right) \right)$.
\noindent\textbf{Unseen} species estimation has a  long history. The first
estimator for the number of new symbols was by Good and
Toulmin~\cite{GoodT56}. The GT estimator is able to solve the problem
as long as $m\le O(n)$, for a constant $\dst$. In other words, after
observing $n$ samples, we can predict how many new symbols we will
observe when we pick up to $n$ new samples. Efron and
Thisted~\cite{EfronT86} proposed an estimator, which is widely used in
practice, although until recently no provable guarantees were
established for its performance or that of any related estimator when
$\msmp > n$. Recently,~\cite{ZouVV16, OrlitskySW16} showed that it is
possible to estimate the number of new symbols up to $m\le n\log
n$. Moreover, the dependence on $n$ is optimal. This result shows that
we can predict the number of new symbols in a much larger sample than
the number of samples we observe. More precisely,~\cite{OrlitskySW16}
showed that it is possible to predict the number of unseen symbols up
to $m = \cO\Paren{\nsmp\log \nsmp /\log (1/\dst)}$.
\noindent\textbf{Distance to uniformity} is related to the problem of testing whether a distribution $\dP$ is uniform or if $\dTV{\dP}{u}\ge 0$. The sample complexity of testing uniformity has been shown to be $\Theta\left(\frac{\sqrt\absz}{\dst^2}\right)$~\cite{BatuFFKRW01,Paninski08}. For estimating $\dTV{\dP}{\dU}$ to an additive $\pm \dst$,~\cite{Valiant11a} showed an estimator with sample complexity $\cO\left(\frac{\absz}{\log \absz}\frac1{\dst^2}\right)$, which has optimal dependence on $\absz$ for a constant $\dst$.
}\ignore{
\subsection{Our results}
Despite all the previous results, it remains unknown if there is a
plug-in estimation routine that can be used to estimate any property
with near-linear number of samples.
In this paper we study a different estimator, \emph{pattern maximum
likelihood (PML)}~\cite{OrlitskySVZ04}, and show that using the PML
plug-in estimator instead of standard ML yields property estimates
that are provably as good or better than the best techniques developed
for the all the problems stated in the above problem.  Furthermore, we
show that the empirical estimator is near-optimal for (up to constant
factors) for estimating support size, unseen species estimation,
entropy, earth mover's distance, distance between pairs of
distributions.
Orlitsky, Santhanam, Viswanathan,~\cite{OrlitskySVZ04, OrlitskySVZ05}
and Zhang proposed a very generic maximum likelihood approach, called
as PML (standing for the pattern/profile maximum likelihood). Instead
of find the SML distribution, which maximizes the probability of the
entire sequence of samples, they proposed to maximize the probability
of the profile (a function) of the samples (see Definition ???). In
addition to several theoretical results, they also show empirically
such estimators have a good performance when the number of samples is
small. They showed that PML can be used to compress profiles with
sub-linear redundancy. However, no guarantees on theoretical
performance of PML for symmetric property estimation is known so
far. In particular, no sample-complexity guarantees are known.
We propose a unified approach to estimating distribution properties
based on PML distribution. We show that instead of designing new
estimators, there is an approach which yields optimal performance for
many symmetric properties.  Our main result, is the following.
\begin{theorem}[Informal]
There is a unified approach based on PML distribution that achieves
the optimal sample complexity for all the four problems mentioned
above (entropy, support, unseen, and distance to uniformity).
\end{theorem}
\paragraph{Computation and approximation.}
On the flip side, we are not aware of efficient algorithms that can
compute the PML distribution exactly. PML for small string lengths
were computed by~\cite{PanAO09}. Recently drawing connections to
permanents a Bethe approximation algorithm was proposed
by~\cite{Vontobel12, Vontobel14}.
Our main result is in fact stronger that plugging in the PML
distribution. The same sample complexity guarantees hold for
approximate PML distributions too. We hope that these strong sample
complexity guarantees will motivate algorithm designers to study the
problem of approximating the PML.
}

%% file: problem-statements.tex
\section{Formal definitions and results}
\label{sec:results}
{Recall that $\Dk$ is the set of all discrete distributions with support at most $k$, and
  $\Delta=\Delta_{\infty}$ is the set of all discrete distributions. 
A property estimator is a mapping $\hat{f} :\cX^n \to \RR$ 
that converts observed samples over $\ab$ to an estimated property value.
The \emph{sample complexity} of $\hat f$ when estimating
a property $f:\Delta\to\RR$ for distributions in 
a collection $\cP\subseteq\Delta$, 
is the number of samples $\hat f$ needs to 
determine $f$ with high accuracy and probability 
for all distributions in $\cP$.
Specifically, for approximation accuracy $\dst$ and 
confidence probability $\delta$, 
\[
C^{\hat{f}}(f, \cP, \delta, \dst)
\ed
\min\left\{n :
p(|f(p)-\hat{f}(X^n)| \geq \dst) \leq \delta\ \forall p \in \cP \right\}.
\]
The sample complexity of estimating $f$ is the lowest sample
complexity of any estimator,
\[
C^*(f, \cP, \delta, \dst) = \min_{\hat{f}}
C^{\hat{f}}(f, \cP, \delta, \dst).
\]

\ignore{
\begin{center}
\fbox{
\begin{minipage}{5.3in}
\smallskip
{\bf Input:} $\beta>0$, sequence $X^n$, $\cP$, symmetric function $\propp{\cdot}$
\begin{enumerate}
\medskip
\item
Compute a $\beta$-approximate PML of $\profile(X^n)$
\medskip
\item
Output $\propp{\mP \profile}$. 
\smallskip
\end{enumerate} 
\end{minipage}
}
\end{center}
}
In the past, different sophisticated estimators were
used for every property in Table~\ref{tab:results}. 
We show that the simple plug-in estimator that uses
any PML approximation $\tdP$, has optimal performance
guarantees for all these properties. 

It can be shown that the sample complexity has only moderate
dependence on $\delta$, that is typically de-emphasized. 
For simplicity, we therefore abbreviate $C^{\hat{f}}(f, \cP,
1/3, \dst )$ by $C^{\hat{f}}(f, \cP,\dst)$.

\ignore{
For \emph{any finite} number of samples, there are distributions that has many elements with \emph{negligible} positive probability that will not appear in our sample. To circumvent this,~\cite{RaskhodnikovaRSS09} considered distributions in $\Dk$ with non-zero probabilities at least $\frac1\absz$, 
$
\Dgk \ed\left\{\dP\in\Dk:
\dPsmb\in\{0\}\cup\left[\frac1k,1\right]\right\}.
$
}
In the next theorem, assume $n$ is at least the optimal sample complexity of estimating entropy, support, support coverage, and distance to uniformity (given in Table~\ref{tab:results}) respectively. 
\begin{theorem}
\label{thm:entropy} 
For all $\dst>c/n^{0.2}$, 
any plug-in $\exp{(-\sqrt n)}$-approximate PML $\tilde{p}$ satisfies,
\begin{description}
\item[Entropy]
\[
C^{\tdP} (H(p), \Delta_k, \dst) \asymp C^* (H(p), \Delta_k,
\dst),\footnote{For  $a,b > 0$, denote $a \lesssim
  b$ or $b \gtrsim a$ if for some universal constant $c$,
  $a/b \leq c $. Denote $a \asymp b$ if
  both $a \lesssim b$ and $a \gtrsim b$.}
\]
\item[Support size]
  \[
C^{\tdP} (S(p)/k, \Dgk, \dst) \asymp C^* (S(p)/k, \Dgk, \dst),
  \]
\item[Support coverage]
    \[
C^{\tdP} (S_m(p)/m, \Delta, \dst) \asymp C^* (S_m(p)/m, \Delta, \dst),
\]
\item[Distance to uniformity]
  \[
  C^{\tdP} (\lVert p-u\rVert_1, \Delta_\cX, \dst) \asymp C^* (\lVert p-u\rVert_1, \Delta_k,
  \dst).
  \]
\end{description}
\end{theorem}

%% file: definitions-notations.tex
\ignore{
In the typical setting for these questions, an unknown underlying
distribution $p\in\Delta_k$ generates $n$ independent samples $X^n\ed
X_1,\upto X_n$, and from this \emph{sample} we would like to estimate
a given property $f(p)$.
Each of these problems was addressed in several papers, and several estimators were derived for each. 
An appealing unified technique for estimating every distribution property
$f(p)$ is the following two step approach. {\bf Step 1: } Find an
approximation $q$ of $p$. {\bf Step 2:} Use the \emph{plug-in} 
estimate $f(q)$ to approximate $f(p)$. Perhaps the simplest estimate
for $p$ is the empirical \emph{maximum likelihood (ML)}. But as
shown in the next Section, while the ML plug-in estimator performs well in the
limit of many samples, it falls short of the best property estimates.
In this paper we study a different estimator, \emph{profile
maximum likelihood (PML)}~\cite{OrlitskySVZ04}, and show that using
the PML plug-in estimator instead of standard ML 
yields property estimates that are provably as good or better
than 
the best techniques developed for the three problems above.
%
%
In addition to deriving a single, uniformly-optimal estimator for
all three problems, another important result we derive, that also plays an essential
part in the proofs, is showing that all three
estimates can be modified to concentrate exponentially fast. 
An obvious question arising is whether exponential concentration and PML optimality hold for additional symmetric properties as well.
}

\section{PML: Profile maximum likelihood}
\label{sec:pml-intro}
\ignore{\subsection{SML: Sequence Maximum likelihood}
Given $n$ independent samples $X^n \ed X_1, X_2, \ldots X_n$ from an
unknown distribution $\dP \in \cP$, the \emph{sequence maximum likelihood
(SML)} estimate assigns the sample the distribution maximizing its
probability,
\[
\maxp{X^n}
\ed
\arg \max_{p \in \cP} p(X^n).
\]
The multiplicity $\Mltsmb$ of a symbol $\smb$ is the number of times
it appears in the sequence $X^n$. 
 It can be easily shown that the
SML estimate is exactly the \emph{empirical frequency}
 estimator
that assigns to each symbol the fraction of times it appears
 in the
sample, $\maxp{X^n}(x)= \frac{\Mltsmb}{n}.$
For example, if $n=11$, and $X^n={a~b~r~a~c~a~d~a~b~r~a}$, $\Mlts{a} =
5, \Mlts{b} =2, \Mlts{c} = 1,\Mlts{d} = 1$, and $ \Mlts{r} = 2$, and
$\maxp{X^{11}}(a)=5/11$, $\maxp{X^{11}}(b)=2/11$,
$\maxp{X^{11}}(c)=1/11$, $\maxp{X^{11}}(d)=1/11$, and
$\maxp{X^{11}}(r)=2/11$.
The SML estimate has many desirable properties such as consistency,
and ease of computation. For example, to learn an unknown $\dP$ over
$k$ elements to an $\ell_1$ distance $\dst$ SML requires $n
= \cO(\absz/\dst^2)$ samples, which is optimal up to constant
factors~\cite{DevroyeL01}.
However, for estimating properties of distributions, especially in the
regime when the number of samples \emph{sublinear} in the domain size,
SML approach is sub-optimal. 
\begin{example}
Consider the uniform distribution $\dU$ over $k$ elements. Suppose we
sample it $n=k/2$ times. Since at most $n$ distinct symbols will appear, the empirical distribution will have entropy at most $\log k-1$ bit. However from Table~\ref{tab:results}, for large $k$, only $n=O(k/\log k)$ samples are required to obtain a 1-bit accurate estimate.
\end{example}
SML is suboptimal as it overfits the distribution estimate.  Hence,
instead of maximizing the probability of the observed sequence we
maximize the observed profile, which is a sufficient statistic for
estimating symmetric properties of interest.
}
\subsection{Preliminaries}
For a sequence $X^n$, recall that the \emph{multilplicity}
$\Mltsmb$ is the number of times $\smb$ appears in $X^n$. Discarding, the labels, profile of a sequence~\cite{OrlitskySVZ04} is defined below.
\begin{definition}
\label{def:profile}
The \emph{profile} of a sequence $X^n$, denoted $\profile(X^n)$ is
the multiset of the multiplicities of all the symbols appearing in
$X^n$.  
\end{definition}
 For example,
$\profile( \emph{a~b~r~a~c~a~d~a~b~r~a}) = \{1, 1, 2, 2, 5\}$,
denoting that there are two symbols appearing once, two appearing
twice, and one symbol appearing five times, removing the association
of the individual symbols with the multiplicities. Profiles are also referred to as histogram order statistics~\cite{paninski2003}, fingerprints~\cite{Valiant11b}, and as histograms of histograms~\cite{BatuFRSW00}.

Let $\Phi^n$ be all profiles of length-$n$ sequences. Then, $\Phi^4 = \{\{1, 1, 1, 1\}, \{1,1, 2\}, \{1, 3\}, \{2, 2\}, \{4\}\}$. In particular, a profile of a length-$n$ sequence is an unordered partition of $n$. Therefore, $|\Phi^n|$,  the number of profiles of length-$n$ sequences is equal to the partition number of $n$. Then, by the Hardy-Ramanujam bounds on the partition number,
\begin{lemma}[\cite{HardyR18, OrlitskySVZ04}]\label{lem:num-profiles}
$|\Phi^n|\le \exp(3\sqrt{\nsmp})$. 
\end{lemma}


For a distribution $\dP$, the probability of a profile $\profile$ is defined as 
\[
\dP(\profile) \ed \sum_{X^n:\profile(X^n) = \profile} \dP(X^n), 
\] 
the probability of observing a sequence with profile $\profile$. Under \emph{i.i.d.} sampling,
\[
\dP(\profile) = \sum_{X^n:\profile(X^n) = \profile} \prod_{i=1}^\nsmp \dP(X_i).
\]
For example, the probability of observing a sequence with profile $\profile = \{1,2\}$ is the probability of observing a sequence with one symbol appearing once, and one symbol appearing twice. A sequence with a symbol $\smb$ appearing twice and $\smby$ appearing once ($e.g.$, $x\ y\ x$) has probability $\dPsmb^2\dPsmby$. Appropriately normalized, for any $\dP$, the probability of the profile $\{1,2\}$ is
\begin{align}
\dP(\{1,2\}) =  \sum_{X^n:\profile(X^n) = \{1,2\}}
\prod_{i=1}^\nsmp \dP(X_i) = {3 \choose 1}  \sum_{a\ne b\in\ab}  p(a)^2p(b),\label{eqn:profprob}
\end{align}
where the normalization factor is independent of $\dP$. The summation is a {monomial symmetric polynomial} in the probability values. See~\cite[Section 2.1.2]{PanT12} for more examples. 
\subsection{Algorithm}
Recall that $\maxp{X^n}$ is the distribution maximizing the
probability of $X^n$. Similarly, define~\cite{OrlitskySVZ04}: 
\[
\maxp{\profile} \ed \max_{\dP\in\cP} \dP(\profile)
\]
as the distribution in $\cP$ that maximizes the probability of
observing a sequence with profile $\profile$.

For example, for $\profile = \{1, 2\}$. For $\cP = \Delta_k$, from~\eqref{eqn:profprob},
\[
\dP_{\profile} = \arg\max_{\dP\in\Dk} \sum_{a\ne b}{\dP(a)^2\dP(b)}.
\]
Note that in contrast, SML only maximizes one term of this expression. 

\noindent We give two examples from the table in~\cite{OrlitskySVZ04} to
distinguish between SML and PML distributions, and also show an instance where
PML outputs distributions over a larger domain than those appearing in
the sample.
\begin{example}
Let $\cX = \{a, b, \ldots, z\}$. Suppose $X^n = x\ y\ x$, then the SML distribution is $(2/3, 1/3)$. However, the distribution in $\Delta$ that maximizes the probability of the profile $\profile(x\ y\ x) = \{1,2\}$ is $(1/2, 1/2)$. Another example, illustrating the power of PML to predict new symbols is $X^n = a\ b\ a\ c$, with profile $\profile(a\ b\ a\ c) = \{1, 1, 2\}$. The SML distribution is $(1/2, 1/4, 1/4)$, but the PML is a uniform distribution over 5 elements, namely $(1/5, 1/5, 1/5,1/5, 1/5)$.
\end{example}

\ignore{The {probability multi-set}
of a distribution $p$ over $\cX$ is the multi-set
$\multiset{p} \ed \sets{p(x):x\in\cX}$ of probabilities, ignoring the
association between elements of $\cX$ and their probability.  For
example, the probability multi-set of the distribution
$(.1,.2,.3,.1,.2,.1)$ over $\sets{1\upto 6}$ is $\sets{.3, .2, .2, .1,
.1, .1}$. Note that, as above, probability multi-set can be identified
with its sorted version.  A symmetric property of a distribution $\dP$
is uniquely determined by $\multiset{\dP}$.}

\ignore{
\subsection{Maximum likelihood}
One of the simplest, yet among the most useful estimators is the
\emph{maximum likelihood (ML)} estimate that assigns to any sample
the distribution maximizing its probability,
\[
\maxp{X^n}
\ed
\arg \max_{p \in \Delta_k} p(X^n).
\]
The multiplicity $\Mltsmb$ of a symbol $\smb$ is the number of times it appears in the sequence $X^n$. 
It can be easily shown that the ML estimate is exactly the \emph{empirical frequency}
estimator that assigns to each symbol the fraction of times it appears
in the sample,
\[
\maxp{X^n}(x)= \frac{\Mltsmb}{n}.
\]
For example, if $n=5$ and $X^5=\text{apple}$, then
$\maxp{X^5}(a)=\maxp{X^5}(l)=\maxp{X^5}(e)=1/5$ and $\maxp{X^5}(p)=2/5$. 
While for $n=6$ and $X^6=\text{google}$, $\maxp{X^5}(g)=\maxp{X^5}(o)=2/6$, and 
$\maxp{X^5}(l)=\maxp{X^5}(e)=1/6$.
It is well known that the ML estimate is \emph{consistent}, namely, approaches $p$ as $n\to\infty$.
Specifically, as $n$ increases, with high probability,
\[
||p -\maxp{X^n}||_1
\ed
\sum_{x} |p(x) - \maxp{X^n}(x)|
\le
\cO\left(\sqrt{\frac{k}{n}}\right).
\]
Hence with high probability, ML learns $p$ to $L_1$ distance $\dst$
with  $\cO(k/\dst^2)$ samples. It is also well known that up to a constant
factor, these many samples are needed by any technique, not just
ML~\cite{DevroyeL01}. 
This near-optimal approximation achieved by the ML distribution
estimator may lead one to hope that the ML plug-in estimator
will perform near-optimally for property estimation as well.
However, the next section show that this is not the case.
}

Suppose we want to estimate a symmetric property $f(\dP)$ of an
unknown distribution $\dP\in\cP$ given $n$ independent samples. Our
high level approach using PML is described below.

\medskip

\begin{center}
\fbox{
\begin{minipage}{5.3in}
\smallskip
{\bf Input:} $\cP$, symmetric function $\propp{\cdot}$, sample $X^n$
\begin{enumerate}
\medskip
\item
Compute $\mP \profile: \arg \max_{p \in \cP} p(\profile(X^n))$.
\medskip
\item
Output $\propp{\mP \profile}$. 
\smallskip
\end{enumerate} 
\end{minipage}
}
\end{center}
\medskip

There are a few advantages of this approach (as is true with any
plug-in approach): $(i)$ the computation of PML is agnostic to the
function $f$ at hand, $(ii)$ there are no parameters to be tuned,
$(iii)$ techniques such as Poisson sampling or median tricks are not
necessary, $(iv)$ well motivated by the maximum-likelihood principle.

{We remark that various aspects of PML have been studied.~\cite{OSVZ11} has a comprehensive collection of various results about PML.~\cite{orlitsky2004universal, OrlitskySZ03} study universal compression and probability estimation using PML distributions.~\cite{OSVZ11, PanAO09, orlitsky2009maximum} derive PML distribution for various special, and small length profiles.~\cite{OSVZ11, anevski2013estimating} prove consistency of PML. Subsequent to the first draft of the current work,~\cite{vatedka2016pattern} showed that both theoretically and empirically plug-in estimators obtained from the PML estimate yield good estimates for symmetric functionals of Markov distributions. }

\paragraph{Comparision to the linear-programming plug-in estimator~\cite{Valiant11b}.}
Our approach is perhaps closest in flavor to the plug-in estimator
of~\cite{Valiant11b}. Indeed, as mentioned in~\cite[Section 2.3]{Valiant12}, their linear-programming estimator is motivated by the question of estimating the PML.
Their result was the first estimator to provide sample complexity bounds in terms of the alphabet size, and accuracy the problems of entropy and support
estimation. Before we explain the differences of the two approaches,
we briefly explain their approach. Define, $\profile_\mlt(X^n)$ to be
the number of elements that appear $\mlt$ times. For example, when
$X^n = {a~b~r~a~c~a~d~a~b~r~a}$, $\profile_1=2, \profile_2=2,$ and
$\profile_5=1$.~\cite{Valiant11b} design a linear program that uses SML for high values of $\mlt$, and formulate a linear program to find a
distribution for which $\EE[\profile_\mlt]$'s are \emph{close} to the
observed $\profile_\mlt$'s. They then plug-in this estimate to estimate the
property. On the other hand, our approach, by the nature of ML principle, tries to find the distribution that best explains the entire profile of the observed data, not just some partial characteristics.  It therefore has the potential to estimate any symmetric property and estimate the distribution closely in any distance measures, competitive with the best possible. For example, the guarantees of the linear program approach are sub-optimal in terms of the desired accuracy $\dst$. For entropy estimation the optimal
dependence is $\frac1\dst$, whereas~\cite{Valiant11b} yields
$\frac1{\dst^2}$. This is more prominent for support size and support
coverage, which have optimal dependence of
$\mathrm{polylog}(\frac1\dst)$, whereas~\cite{Valiant11b} gives a
$\frac1{\dst^2}$ dependence. Besides, we analyze the first method
proposed for estimating symmetric properties, designed from the first
principles, and show that in fact it is competitive with the optimal
estimators for various problems.


%% file: outline.tex
\section{Proof outline}
\label{sec:outline}

Our arguments have two components. In Section~\ref{sec:pml_maximum} we prove a general result for the performance of plug-in estimation via maximum likelihood approaches.  

Let $\cP$ be a class of distributions over $\cZ$, and $f:\cP\to\RR$ be a function. For $z\in\cZ$, let
\[
\mP
z\ed\arg \max_{\dP\in\cP}\dPsmbz
\]
be the maximum-likelihood estimator of $z$ in $\cP$. Upon observing $z$, $f(\mP z)$ is  the ML estimator of $f$. In Theorem~\ref{thm:pml_property}, we show that if there is an estimator that achieves error probability $\delta$, then the ML estimator has an error probability at most $\delta|\cZ|$. We note that variations of this result in the asymptotic statistics were studied before (see~\cite{LehmannC98}).  Our contribution is to use these results in the context of symmetric properties and show sample complexity bounds in the non-asymptotic regime.

We \emph{emphasize that}, throughout this paper $\cZ$ will be the set of profiles of length $n$, and $\cP$ will be distributions induced over profiles by length-$n$ $i.i.d.$ samples. Therefore, we have $|\cZ|=|\Phi^n|$. By
Lemma~\ref{lem:num-profiles}, if there is a \emph{profile based} estimator with error probability $\delta$, then the PML approach will have error probability at
most $\delta\exp(3\sqrt n)$. Such arguments were used in hypothesis testing to show the existence of competitive testing algorithms for fundamental statistical problems~\cite{AcharyaDJOP11, AcharyaDJOPS12}.

At its face value this seems like a weak result. Our second key step
is to prove that for the properties we are interested, it is 
possible to obtain very sharp guarantees. For example, we show that if
we can estimate the entropy to an accuracy $\pm\dst$ with error
probability $1/3$ using $n$ samples, then we can estimate the
entropy to accuracy $\pm2\dst$ with error probability
$\exp(-n^{0.9})$ using only $2n$ samples. Using this sharp
concentration, the new error probability term dominates
$|\Phi^n|$, and we obtain our results. The arguments for
sharp concentration are based on modifications to existing
estimators and a new analysis. Most of these results are technical and
are in the appendix.

%% file: property-by-pml-2.tex
\section{Estimating properties via maximum likelihood}
\label{sec:pml_maximum}
In this section, we prove the performance guarantees of ML property estimation in a general set-up. Recall that $\cP$ is a collection of distributions over $\cZ$, and $f:\cP\to\RR$. Given a sample $Z$ from an unknown $\dP\in\cP$, we want to estimate $\propp{\dP}$. The maximum likelihood approach is the following two-step procedure.
\medskip
\begin{enumerate}
\item Find $\mP Z=\arg \max_{\dP\in\cP}\dPsmbZ$.
\item Output $f(\mP Z)$. 
\end{enumerate}

\medskip


\noindent We bound the performance of this approach in the following theorem.
\begin{theorem}
\label{thm:pml-property-exact}
Suppose there is an estimator $\fhat:\cZ\to\RR$, such that for any $\dP$, and $Z\sim\dP$, 
\begin{align}
\probof{\absv{\propp \dP
 - \hpropz Z}>\dst}< \err,\label{eqn:est-exists}
\end{align}
then 
\begin{align}
\probof{\absv{\propp \dP - \propp{\mP Z}}>2\dst}\le 
\err\cdot\absv{\cZ}.\label{eqn:exact-ml}
\end{align}
\end{theorem}

\noindent \textit{Proof.}
Consider symbols with $\dPsmbz \ge \err$ and $\dPsmbz < \err$
separately. A distribution $\dP$ with $\dPsmbz\ge\err$ outputs $z$ with probability at least $\err$. For~\eqref{eqn:est-exists} to hold, we must have, 
$
\absv{\propp {\dP} - \hpropz z}<\dst.
$
By the definition of ML, $\mP z(z) \ge \dPsmbz\ge\err$, and again for~\eqref{eqn:est-exists} to hold for $\mP z$, $\absv{\propp {\mP z} - \hpropz z}<\dst$. By the triangle inequality, for all such $\smbz$, 
\begin{align*}
\absv{\propp {\dP} -\propp {\mP z}}\le \absv{\propp {\dP} - \hpropz
    z}+ \absv{\propp {\mP z} - \hpropz z} \le 2\dst.
\end{align*}
Thus if $p(z) \geq \err$, then PML satisfies the required guarantee
with zero probability of error, 
and any error occurs only when  $\dP(z)<\err$. We bound this
probability as follows. When $Z\sim\dP$, 
\[
\probof{\dPsmbZ<\err} \le \sum_{\smbz\in\cZ:\dPsmbz<\err}\dPsmbz <
\err\cdot\absv{\cZ}. \eqed
\]
For some problems, it might be easier to just approximate the ML, instead of finding it exactly. We define an approximation ML as follows:
\begin{definition}[$\beta$-approximate ML]
\label{def:appr}
Let $\beta\le1$. For $Z\in\cZ$, $\tdP_Z\in\cP$ is a $\beta$-approximate ML distribution if
$\tdP_z (z)\ge \beta\cdot \mP z(z)$. When $\cZ$ is profiles of length-$n$, a $\beta$-approximate PML is a distribution $\tdP_\profile$ such that $\tdP_\profile(\profile)\ge\beta\cdot \dP_{\profile}(\profile)$. 
\end{definition}

The next result proves guarantees for any $\beta$-approximate ML estimator.
\begin{theorem}
\label{thm:pml_property}
Suppose there exists an estimator satisfying~\eqref{eqn:est-exists}. For any $\dP\in\cP$ and $Z\sim\dP$, any $\beta$-approximate ML $\tdP_Z$ satisfies:
\begin{align}
\probof{\absv{\propp \dP - \propp{\tilde{p}_Z}}>2\dst} \leq 
\frac{\err\cdot\absv{\cZ}}{\beta}.\nonumber
\end{align} 
\end{theorem}
The proof is very similar to the previous theorem and is presented in
the Appendix~\ref{app:approx-ml}.

\subsection{Competitiveness of ML estimators via median trick}

Suppose for a property $f(p)$, there is an estimator with sample
complexity $n$ that achieves an accuracy $\pm\dst$ with probability of
error at most $1/3$.  The standard method to boost the error
probability is the median trick: (i) Obtain $O(\log(1/\delta))$
independent estimates using $O(n\log(1/\delta))$ independent
samples. (ii) Output the median of these estimates. This is an
$\dst$-accurate estimator of $f(p)$ with error probability at most
$\delta$.}  By definition, estimators are a mapping from the samples
  to $\RR$. However, in many applications the estimators map from a
  much smaller (some sufficient statistic) of the samples. Denote by
  $Z_n$ the space consisting of all sufficient statistics that the
  estimator uses. For example, estimators for symmetric properties,
  such as entropy typically use the profile of the sequence, and hence
  $Z_n=\Phi^n$. Using the median-trick, we get the following result.

\begin{corollary}
\label{cor:competitive}
Let $\fhat:Z_n\to \RR$ be an estimator of $f(p)$ with accuracy $\dst$
and error-probability $1/3$. The ML estimator achieves accuracy
$2\dst$ using
\[
\min \left\{n': \frac{n'}{20\log(3Z_{n'})}\right\}  > n.
\]	
\end{corollary}
\begin{proof}
Since $n$ is the number of samples to get an error probability $1/3$,
by the Chernoff bound, the error after $n'$ samples is at most $\exp(-
(n'/(20n)))$. Therefore, the error probability of the ML estimator for
accuracy $2\dst$ is at most $\exp(- (n'/(20n)))Z_{n'}$, which we desire
to be at most $1/3$.
\end{proof}

For estimators that use the profile of sequences,
$|\Phi^n|<\exp(3\sqrt{n})$. Plugging this in the previous result shows
that the PML based approach has a sample complexity of at most
$\cO(n^2)$. This result holds for all symmetric properties, independent of
$\dst$, and the alphabet size $k$. For the problems mentioned earlier,
something much better in possible, namely the PML approach is optimal
up to constant factors.

%% file: sharp-concentration.tex
\section{Sample optimality of PML}
\label{sec:approach}
\ignore{
We first prove that PML is \emph{competitive} with respect to the best
estimator for \emph{any symmetric property}. This is a side-result and
a direct corollary of the results in the previous section. We then
provide much sharper concentration bounds for estimating the
properties we are considering, and use it to prove the optimality of
PML.
\subsection{Median trick and competitiveness of PML}
Suppose for a property $f(p)$, there is an estimator with sample complexity $n$ that achieves an accuracy $\pm\dst$ with probability of error at most $1/3$. 
The standard method to boost the error probability is the median trick: (i) Obtain $O(\log(1/\delta))$ independent estimates using $O(n\log(1/\delta))$ independent samples. (ii) Output the median of these estimates. This is an $\dst$-accurate estimator of $f(p)$ with error probability at most $\delta$.
By definition, estimators are a mapping from the samples to $\RR$. However, in many applications the estimators map from a much smaller (some sufficient statistic) of the samples. Denote by $Z_n$ the space consisting of all sufficient statistics that the estimator uses. For example, estimators for symmetric properties, such as entropy typically use the profile of the sequence, and hence $Z_n=\Phi^n$. Using the median-trick, we get the following result.
\begin{corollary}
\label{cor:competitive}
Let $\fhat:Z_n\to \RR$ be an estimator of $f(p)$ with accuracy $\dst$ and error-probability $1/3$. The ML estimator achieves accuracy $2\dst$ using
\[
\min \left\{n': \frac{n'}{\log(3Z_{n'})}\right\}  > n.
\]	
\end{corollary}
\begin{proof}
Since $n$ is the number of samples to get an error probability $1/3$, by the median trick, the error after $n'$ samples is at most $\exp(- O(n'/n))$. Therefore, the error probability of the ML estimator for accuracy $2\dst$ is at most $\exp(- O(n'/n))Z_{n'}$, which we desire to be at most $1/3$.
\end{proof}
For estimators that use the profile of sequences, $|\Phi^n|<\exp(3\sqrt{n})$. Plugging this in the previous result shows that the PML based approach has a sample complexity of at most $n^2$. This result holds for all symmetric properties, independent of $\dst$, and the alphabet size $k$. For the problems mentioned earlier, something much better in possible, namely the PML approach is optimal up to constant factors.
}
\subsection{Sharp concentration for some interesting properties}

To obtain sample-optimality guarantees for PML, we need to drive the error probability down much faster than the median trick. We achieve this by using McDiarmid's inequality stated below.
%
Let $\fhat:\cX^*\to\RR$. Suppose $\fhat$ gets $\nsmp$ independent samples $X^n$ from an unknown distribution. Moreover, changing one of the $X_j$ to any $X_j'$ changed $\fhat$ by at most $c_*$. Then McDiarmid's inequality (bounded difference inequality,~\cite[Theorem 6.2]{BoucheronLM13}) states that,
\begin{align}
\probof{\absv{\fhat(X^n) - \EE[\fhat(X^n)]}>t}\le 2\exp
\left(-\frac{2t^2}{n c^2_*} \right).\label{eqn:mcdia-one}
\end{align}
This inequality can be used to show strong error probability bounds for many problems. We mention a simple application for estimating discrete distributions. 
\begin{example}
It is well known~\cite{DevroyeL01} that SML requires $\Theta(k/\dst^2)$ samples  to estimate $\dP$ in $\ell_1$ distance with probability at least $2/3$. In this case, $\fhat(X^n) = \sum_{\smb}\absv{\frac{\Mltsmb}\nsmp-\dPsmb}$, and therefore $c_*$ is at most $2/n$. Using McDiarmid's inequality, it follows that SML has an error probability of $\delta = 2\exp(-k/2)$, while still using $\Theta (k/\dst^2)$ samples. 
\end{example}


Let $B_n$ be the bias of an estimator $\fhat(X^n)$ of $f(\dP)$, namely 
\begin{align}
B_n\ed\absv{f(\dP)-\EE[\fhat(X^n)]}.\nonumber
\end{align}
By the triangle inequality, 
\begin{align}
\absv{f(\dP)-\fhat(X^n)} \leq \absv{f(\dP)-\EE[\fhat(X^n)]} +\absv{\fhat(X^n)-\EE[\fhat(X^n)]} = B_n + \absv{\fhat(X^n)-\EE[\fhat(X^n)]}. \nonumber 
\end{align}
Plugging this in~\eqref{eqn:mcdia-one},
\begin{align}
\probof{\absv{f(\dP)-\fhat(X^n)]}>t+B_n}\le 2\exp \left(-\frac{2t^2}{n c^2_*} \right).\label{eqn:mcdia-two}
\end{align}

With this in hand, we need to show that $c_*$ can be bounded for estimators for the properties we consider. In particular, we will show that
\begin{lemma}
  \label{lem:prop-of-interest}
 Let $\alpha>0$ be a fixed constant. 
  For entropy, support, support coverage, and distance to uniformity there exist profile based
  estimators that use the optimal number of samples (given in
  Table~\ref{tab:results}), have bias $\dst$ and if we change
  any of the samples, changes by at most $c \cdot \frac{\nsmp^\alpha}{\nsmp}$, where $c$ is a positive constant. 
\end{lemma}
We prove this lemma by proposing several modifications to the existing sample-optimal estimators. The modified estimators will preserve the sample complexity up to constant factors and also have a small $c_*$. The proof details are given in the appendix.

Using~\eqref{eqn:mcdia-two} with Lemma~\ref{lem:prop-of-interest},
\begin{theorem}
Let $n$ be the optimal sample complexity of estimating entropy, support, support coverage and distance to uniformity (given in table~\ref{tab:results}) and $c$ be a large positive constant. Let $\dst \geq c/n^{0.2}$, then any for any $\beta>\exp\left(-\sqrt n\right)$, the $\beta$-PML estimator estimates entropy, support, support coverage, and distance to uniformity to an accuracy of $4\dst$ with probability at least $1-\exp(-\sqrt n)$.\footnote{{The above theorem also works for any $\dst \gtrsim 1/n^{0.25-\eta}$ for any $\eta > 0$}}
\end{theorem}
\begin{proof}
Let $\alpha = 0.1$.  By Lemma~\ref{lem:prop-of-interest}, for each property of interest, there are estimators based on the profiles of the samples such that using near-optimal number of samples, they have bias $\dst$ and maximum change if we change any of the samples is at most $n^{\alpha}/n$. Hence, by McDiarmid's inequality, an accuracy
of $2\dst$ is achieved with probability at least $1-\exp\left(-2\dst^2
n^{1-a}/c^2\right)$. Now suppose $\tdP$ is any $\beta$-approximate
PML distribution. Then by Theorem~\ref{thm:pml_property}
\[
\probof{\absv{f(\dP)-f(\tdP)}>4\dst}< \frac{\delta
  \cdot|\Phi^n|}{\beta} \le \frac{\exp(-2\dst^2n^{1-a}/c^2)\exp(3\sqrt
  n)}{\beta}\le \exp(-\sqrt n),
\]
where in the last step we used $\dst^2 n^{1-a}\gtrsim c'\sqrt n$, and
$\beta>\exp(-\sqrt \nsmp)$.
\end{proof}

\section{Discussion and future directions}
We studied estimation of symmetric properties of discrete
  distributions using the principle of maximum likelihood, and proved
  optimality of this approach for a number of problems. A number of
  directions are of interest.  We believe that the lower bound
  requirement on $\dst$ is perhaps an artifact of our proof technique,
  and that the PML based approach is indeed optimal for all ranges of
  $\dst$.
Approximation algorithms for estimating the PML
  distributions would be a fruitful direction to pursue. Given our
  results, approximations stronger than $\exp(-\dst^2 n)$ would be
  very interesting. In the particular case when the desired accuracy
  is a constant, even an exponential approximation would be sufficient
  for many properties.  We plan to apply the heuristics proposed
  by~\cite{Vontobel12} for various problems we consider, and compare
  with the state of the art provable methods. 

%
%

%% file: acknowledgements.tex
\section*{Acknowledgements}
We thank Ashkan Jafarpour, Shashank Vatedka, Krishnamurthy Viswanathan, and 
Pascal Vontobel for very helpful discussions and suggestions.   

%% file: unseen.tex
\section{Support and support coverage}
We analyze both support coverage and the support estimation via a single
approach. We first start with support coverage. Recall that the goal
is to estimate $S_m(p)$, the expected number of distinct symbols that
we see after observing $m$ samples from $p$.  By the linearity of
expectation,
\[
S_m(p) = \sum_{x \in \cX} \EE[\II_{N_x(X^m) > 0}] = \sum_{x \in \cX} \left( 1 - (1-p(x))^m\right).
\]
The problem is closely related to the support coverage problem~\cite{OrlitskySW16},
where the goal is to estimate $U_{t}(X^n)$, the number of new distinct
symbols that we observe in $n \cdot t$ additional samples. Hence
\[
S_{m}(p) = \EE\left[ \sum^n_{i=1} \prv_i\right] + \EE[U_t],
\]
where $t = (m-n)/n$.  We show that the modification of an estimator
in~\cite{OrlitskySW16} is also near-optimal and satisfies conditions in
Lemma~\ref{lem:prop-of-interest}.  We propose to use the following estimator
\[
\hat{S}_m(p) = \sum^n_{i=1} \prv_i + \sum^n_{i=1} \prv_{i} (-t)^i \Pr(Z \geq i),
\]
where $Z$ is a Poisson random variable with mean $r$ and $t
=(m-n)/n$. We remark that the proof also holds for Binomial smoothed
random variables as discussed in~\cite{OrlitskySW16}.

We need to bound the maximum coefficient and the bias to apply
Lemma~\ref{lem:prop-of-interest}. We first bound the maximum coefficient of this
estimator.
\begin{lemma}
  \label{lem:unco}
For all $n \leq m/2$, the maximum coefficient of $\hat{S}_m(p)$ is
at most 
\[
1 + e^{r(t-1)}.
\]
\end{lemma}
\begin{proof}
For any $i$, the coefficient of $\prv_i$ is 
\[
1 + (-t)^{i} \Pr(Z \geq i).
\]
It can be upper bounded as 
\[
1 + \sum^t_{i=0} \frac{e^{-r} (rt)^i}{i!} =
1 + e^{r(t-1)}.
\]
\end{proof}
The next lemma bounds the bias of the estimator.
\begin{lemma}
  \label{lem:unbias}
For all $n \leq m/2$, the bias of the estimator is bounded by 
\[
|\EE[\hat{S}_m(p)] - S_{m}(p)| \leq 2 + 2 e^{r(t-1)} + \min(m, S(p)) e^{-r}.
\]
\end{lemma}
\begin{proof}
As before let $t = (m-n)/n$. 
\begin{align*}
\EE[\hat{S}_m(p)] - S_{m}(p)
& =\sum^n_{i=1} \EE[\prv_i] + \EE[U^{\text{SGT}}_{t}(X^n)]  - \sum_{x \in \cX} \left( 1 - (1-p(x))^m\right) \\
& = \EE[U^{\text{SGT}}_{t}(X^n)]  - \sum_{x \in \cX} \left( (1 - p(x))^n - (1-p(x))^m\right).
\end{align*}
Hence by Lemma $8$ and Corollary $2$, in \cite{OrlitskySW16}, we get
\[
|\EE[\hat{S}_m(p)] - S_{m}(p)| \leq 2 + 2 e^{r(t-1)} + \min(m, S(p)) e^{-r}.
\]
\end{proof}
Using the above two lemmas we prove results for both the observed
support coverage and support estimator.

\subsection{Support coverage estimator}
Recall that the quantity of interest in support coverage estimation is
$S_{m}(p)/m$, which we wish to estimate to an accuracy of $\dst$.
\begin{proof}[Proof of Lemma~\ref{lem:prop-of-interest} for observed]
If we choose $r = \log \frac{3}{\dst}$, then by
Lemma~\ref{lem:unco}, the maximum coefficient of $\hat{S}_m(p)/m$ is
at most
\[
\frac{2}{m} e^{\frac{m}{n}\log \frac{3}{\dst}},
\]
which for $m \leq \alpha \frac{n \log (n/2^{1/\alpha})}{\log
  (3/\dst)}$ is at most $n^{\alpha}/m < n^{\alpha}/n$.
Similarly, by Lemma~\ref{lem:unbias},
\[
\frac{1}{m}|\EE[\hat{S}_m(p)] - S_{m}(p)| \leq \frac{1}{m} (2 + 2
e^{r(t-1)} + m e^{-r}) \leq \dst,
\]
for all $\dst > 6n^{\alpha}/n $.
\end{proof}

\subsection{Support estimator}
Recall that the quantity of interest in support estimation is
$S(p)/k$, which we wish to estimate to an accuracy of $\dst$.
\begin{proof}[Proof of Lemma~\ref{lem:prop-of-interest} for support]
Note that we are interested in distributions with all the non zero
probabilities are at least $1/k$. We propose to estimate $S(p)/k$ using
\[
\hat{S}_m(p)/k,
\]
for $m = k \log \frac{3}{\dst}$. Note that for this choice of $m$
\[
0 \leq S(p)-S_m(p) = \sum_{x} (1 -  ( 1 - (1-p(x))^m)) = \sum_{x} (1-p(x))^m \leq \sum_{x} e^{-mp(x)} \leq k e^{-\log\frac{3}{\dst}} = \frac{k\dst}{3}.
\]
If we choose  $r = \log \frac{3}{\dst}$, then by Lemma~\ref{lem:unco},
the maximum coefficient of $\hat{S}_m(p)/k$ is at most
\[
\frac{2}{k} e^{\frac{m}{n}\log \frac{3}{\dst}},
\]
which for $n \geq \alpha \frac{k}{\log (k/2^{1/\alpha})} \log^2 \frac{3}{\dst}$
is at most $k^{\alpha}/k < n^{\alpha}/n$.  Similarly, by
Lemma~\ref{lem:unbias},
\begin{align*}
  \frac{1}{k}|\EE[\hat{S}_m(p)] -  S(p)|
  & \leq \frac{1}{k}|\EE[\hat{S}_m(p)] -  S_{m}(p)| +  \frac{1}{k} |S(p)-S_m(p)| \\
  & \leq \frac{1}{k} (2 + 2 e^{r(t-1)} + k e^{-r}) + \frac{\dst}{3} \\
  & \leq \dst,
\end{align*}
for all $\dst > 12 n^{\alpha}/n$. 
\end{proof}

%% file: entropy-l1dist.tex
\section{Entropy and distance to uniformity}
The known optimal estimators for entropy and distance to uniformity
both depend on the best polynomial approximation of the corresponding
functions and the splitting trick~\cite{WuY14a, JiaoVHW15}. Building
on their techniques, we show that a slight modification of their
estimators satisfy conditions in Lemma~\ref{lem:prop-of-interest}.
Both these functions can be written as functionals of the form:
\[
f(\dP) = \sum_\smb g(\dPsmb),
\]
where $g(y) = -y\log y$ for entropy and $g(y)=\absv{y-\frac1\absz}$ for uniformity.

Both\cite{WuY14a, JiaoVHW15} first approximate $g(y)$ with
$P_{L,g}(y)$ polynomial of some degree $L$. Clearly a larger
degree implies a smaller bias/approximation error, but estimating a
higher degree polynomial also implies a larger statistical estimation
error. Therefore, the approach is the following:
\begin{itemize}
\item
For small values of $\dPsmb$, we estimate the polynomial $P_{L,g}(\dPsmb) = \sum^L_{i=1} b_i \cdot (p(x))^i$.
\item
For large values of $\dPsmb$ we simply use the empirical estimator for $g(\dPsmb)$.  
\end{itemize}

However, it is not a priori known which symbols have high probability
and which have low probability. Hence, they both assume that they
receive $2n$ samples from $\dP$. They then divide them into two set of
samples, $X_1^{'}, \ldots, X_n^{'}$, and $X_1, \ldots, X_n$. Let
$\Mltsmb^{'}$, and $\Mltsmb$ be the number of appearances of symbol
$\smb$ in the first and second half respectively. They propose to use
the estimator of the following form:
\[
\hat{g}(X_1^{2\nsmp})
=\max\left\{\min\left\{\sum_x{g_x},f_{\max}\right\}, 0\right\}.
\]
where $f_{\max}$ is the maximum value of the property $f$ and 
\[
g_x=
\begin{cases}
G_{L,g}(\Mltsmb), & \text{ for } \Mltsmb^{'}<c_2\logk, \text{ and }
\Mltsmb< c_1\logk,\\ 0, & \text{ for } \Mltsmb^{'}<c_2\logk, \text{
  and } \Mltsmb\ge c_1\logk,\\ g \left(\frac{\Mltsmb}{n}\right)+ g_n
, & \text{ for } \Mltsmb^{'}\ge c_2\logk,
\end{cases}
\]
where $g_n$ is the first order bias correction term for $g$,
$G_{L,g}(N_x) = \sum^L_{i=1} b_i \flnpwrss{N_x}{i}/\flnpwrss{n}{i}$ is
the unbiased estimator for $P_{L,g}$, and $c_1$ and $c_2$ are two
constants which we decide later.  We remark that unlike previous
works, we set $g_x$ to $0$ for some values of $N_x$ and $N'_x$ to
ensure that $c^*$ is bounded.  The following lemma bounds $c^*$ for
any such estimator $\hat{g}$.
\begin{lemma}
  \label{lem:general_est}
  For any estimator $\hat{g}$ defined as above, changing any one of
  the values changes the estimator by at most
  \[
  8\max\left(e^{L^2/n}\max
  |b_i|, \frac{L_g}{n}, g\left(\frac{c_1 \log (n)}n\right), g_n\right),
  \]
  where $L_g = n \max_{i\in \NN} |g(i/n)- g((i-1)/n)|$.
\end{lemma}
\ignore{
There exist optimal estimators for both entropy and distance to
uniformity based on best polynomial approximation. We provide some key
ingredients of the proof and refer the reader to the details
in~\cite{WuY14a} for a treatment of best polynomial estimation. In
fact a very similar approach also gives sample-optimal estimation of
the distance to uniformity.
The objective in these problems are estimating functionals of the form:
\[
f(\dP) = \sum_\smb g(\dPsmb),
\]
where $g(x) = -x\log x$ for entropy, $g(x)=\absv{x-\frac1\absz}$ for uniformity, and $g(x)=\II\{x>0\}$.
The empirical estimator has a large bias, especially in the regime of insufficient samples. Moreover, Paninski~\cite{paninski2003} had shown that unbiased estimators exist only for polynomial functions $g(x)$, and showed that no unbiased estimator exists for the entropy function. This motivated the following approach~\cite{WuY14a, JiaoVHW15}. We first approximate $g(x)$ with a polynomial of some degree $L$. Clearly a larger degree implies a smaller bias/approximation error, but estimating a higher degree polynomial also implies a larger approximation error. Moreover larger values of $x$ also leads to a high variance. Therefore, the approach is the following:
\begin{itemize}
\item
For small values of $\dPsmb$, we estimate the polynomial $g(\dPsmb)$.
\item
For large values of $\dPsmb$ we simply use the first order bias
corrected empirical estimator for $g(\dPsmb)$.
\end{itemize}
One of the differences with the previous approaches is Poisson-sampling. We do not make the assumption of Poisson sampling, and instead work with a fixed set of samples. This helps us in invoking McDiarmid's inequality. We now describe the method in a greater detail (See~\cite[Section~4]{WuY14a}).
Assume that we receive $2n$ samples from $\dP$. We then divide them
into two set of samples, $X_1^{'}, \ldots, X_n^{'}$, and $X_1, \ldots,
X_n$. Moreover, let $\Mltsmb^{'}$, and $\Mltsmb$ be the number of
appearances of symbol $\smb$ in the first and second half
respectively.
\noindent We need one definition from approximation theory.
\begin{definition}
The \emph{best polynomial approximation} of degree $L$ for a function $g(x)$ over an interval $I$ is a degree $L$ polynomial $g_L(x)$ that minimizes 
\[
\sup_{x\in I}\absv{g_L(x)-g(x)}.
\]
Let $E_{g,L, I}$ be the value of this min-sup. 
\end{definition}
Let $g_{L,H}(x)$ be the best polynomial approximation of $x\log (1/x)$ in the interval $[0, 1]$, and for $-1<c<1$, let $g_{L, U}$ be the best polynomial approximation of $|x-c|$ over $[-1,1]$. Then, from~\cite[Lemma 2]{CaiL11},~\cite[Equation 7.2.2]{Timan63} we have
\begin{itemize}
\item 
The largest coefficients of 	$g_{L,H}$ and $g_{L,U}$ are at most $2^{3L}$. 
\item The approximation errors for these functions are bounded by:
\begin{align}
E_{-x\log x,L, [0,1]}\le O\left(\frac{1}{L^2}\right)\ \text{ and } \ E_{|x-c|,L, [-1,1]}\le O\left(\frac{\sqrt{1-c^2}}{L}\right)\label{eqn:est-guar}.
\end{align}
\end{itemize}
 and $g_{L,U} = \sum_{i=0}^L b_ix^i$
We will use three constants $c_0, c_1, and c_2$. These constants will have the same value for estimating the entropy and distance to uniformity and serve the following purpose. Symbols $\smb$ with $\Mltsmb^{'}$ more than $c_2\logk$ times will use empirical estimators for $g(\dPsmb)$. Elements for which $\Mltsmb^{'}<c_2\logk$, and $\Mltsmb<c_1\logk$, we will use a polynomial estimator. Finally, for symbols with $\Mltsmb^{'}<c_2\logk$, and $\Mltsmb>c_1\logk$ we output $0$.
}
\subsection{Entropy}
The following lemma is adapted from Proposition 4 in~\cite{WuY14a}
where we make the constants explicit.
\begin{lemma}
  Let $g_n = 1/(2n)$ and $\alpha > 0$.  Suppose $c_1=2c_2$, and $c_2>35$,
  Further suppose that
  $n^3\left(\frac{16c_1}{\alpha^2}+\frac1{c_2}\right)>\log k\cdot\log
  n$. There exists a polynomial approximation of $- y \log y$ {with degree $L = 0.25 \alpha$}, over
  $[0, c_1\frac{\logk}{n}]$ such that $\max_{i} |b_i| \leq
  n^{\alpha}/n$ and the bias of the entropy estimator is at most
  $\cO\left(\left(\frac{c_1}{\alpha^2}+\frac1{c_2} + {\frac{1}{n^{3.9}}}\right)\frac{\absz}{n\logk}\right)$.
\end{lemma}
\begin{proof}
Our estimator is similar to that of~\cite{WuY14a, JiaoHW16} except for
the case when $\Mltsmb^{'}<c_2\logk$, and $\Mltsmb>c_1\logk$. For any
$\dPsmb$, and $\Mltsmb^{'}$ and $\Mltsmb$ both distributed
$Bin(n\dPsmb)$. By the Chernoff bounds for binomial distributions, the
probability of this event can be bounded by,
\begin{align}
\max_{\dPsmb}\probof{\Mltsmb^{'}<c_2\logk, \Mltsmb>2c_2\logk}
\le \frac1{n^{0.1 \sqrt 2 c_2}}\le \frac1{n^{4.9}}.\nonumber
\end{align}
Therefore, the additional bias the modification introduces is at most
$k\log k/n^{4.9}$ which is smaller than the bias term
of~\cite{WuY14a, JiaoHW16}.

The largest coefficient can be bounded by using that the best
polynomial approximation of degree $L$ of $x\log x$ in the interval
$[0,1]$ has all coefficients at most $2^{3L}$. Therefore, the largest
change we have (after appropriately normalizing) is the largest value
of $b_i$ which is
\[
\frac{2^{3L}e^{L^2/n}}{n}.
\]
For $L = 0.25\alpha\log n$, this is at most $\frac {n^a}{n}$. 
\end{proof}

The proof of Lemma~\ref{lem:prop-of-interest} for entropy follows from
the above lemma and Lemma~\ref{lem:general_est} and by substituting $n
= \cO\left(\frac{k}{\log k} \frac{1}{\dst}\right)$.

\subsection{Distance to uniformity} 

We state the following result stated in~\cite{JiaoHW16}.
\begin{lemma}

  Let $c_1>2c_2$, $c_2=35$. There is an estimator for distance to uniformity that changes by at most $n^{\alpha}/\nsmp$ when a sample is changed, and the bias of the estimator is at most $O(\frac{1}{\alpha}\sqrt{\frac{c_1\logk}{k\cdot n}})$.
\end{lemma}
\begin{proof}
Estimating the distance to uniformity has two regions based on $N'_x$ and $N_x$. 
\paragraph{Case 1: $\frac1\absz < c_2\logk/n$.} In this case, we use the estimator defined in the last section for $g(x) = |x-1/k|$. 

\paragraph{Case 2: $\frac1\absz > c_2\logk/n$.} In this case, we have a slight change to the conditions under which we use various estimators:
\[
g_x=
\begin{cases}
G_{L,g}(\Mltsmb), & \text{ for }
\absv{\Mltsmb^{'}-\frac1{\absz}}<\sqrt{\frac{c_2\logk}{kn}}, \text{
  and } \absv{\Mltsmb-\frac1{\absz}}<\sqrt{\frac{c_1\logk}{kn}},\\ 0,
& \text{ for }
\absv{\Mltsmb^{'}-\frac1{\absz}}<\sqrt{\frac{c_2\logk}{kn}}, \text{
  and } \absv{\Mltsmb-\frac1{\absz}}\ge\sqrt{\frac{c_1\logk}{kn}},\\ g
\left(\frac{\Mltsmb}{n}\right) , & \text{ for }
\absv{\Mltsmb^{'}-\frac1{\absz}}\ge\sqrt{\frac{c_2\logk}{kn}}.
\end{cases}
\]

The estimator proposed in~\cite{JiaoHW16} is slightly different, assigning $G_{L,g}(\Mltsmb)$ for the first two cases. We design the second case to bound the maximum deviation. The bias of their estimator was shown to be at most $\cO\left(\frac1{L}\sqrt{\frac{\logk}{k\cdot n\logk}}\right)$, which can be shown by using 
\cite[Equation 7.2.2]{Timan63}
\begin{align}
E_{|x-\tau|,L, [0,1]}\le O\left(\frac{\sqrt{\tau(1-\tau)}}{L}\right).\label{eqn:tau-unif}
\end{align}
By our choice of $c_1, c_2$, our modification changes the bias by at most $1/n^4<\dst^2$. 

To bound the largest deviation, we use the fact (\cite[Lemma 2]{CaiL11}) that the largest coefficient of the best degree-$L$ polynomial approximation of $|x|$ in $[-1,1]$ has all coefficients at most $2^{3L}$. Similar argument as with entropy yields that after appropriate normalization, the largest difference in estimation will be at most $n^{\alpha}/n$.
\end{proof}

The proof of Lemma~\ref{lem:prop-of-interest} for entropy follows from
the above lemma and Lemma~\ref{lem:general_est} and by substituting $n
= \cO\left(\frac{k}{\log k} \frac{1}{\dst^2}\right)$.

%% file: proof-approximation.tex
\section{Proof of approximate ML performance}
\label{app:approx-ml}
\noindent \textit{Proof.}
We consider symbols such that $\dPsmbz \ge \err/\beta$ and $\dPsmbz < \err/\beta$
separately. For an $\smbz$ with $\dPsmbz\ge\err/\beta$, by the definition of
$\propp{\mP Z}$, 
\begin{align*}
\tilde{p}_z(z) \geq \mP z(z) \beta \ge p(z) \beta \ge \err. 
\end{align*} 
Applying~\eqref{eqn:est-exists} to $\tilde{p}_z$, we have for $Z\sim\tilde{p}_z$, 
\begin{align*}
\err >
\probof{\absv{\propp {\tilde{p}_z} - \hpropz Z}>\dst}
\ge\ \tilde{p}_z(z) \cdot\II\left\{\absv{\propp {\tilde{p}_z} - \hpropz z}>\dst\right\}
\ge\ \err\cdot \II\left\{\absv{\propp {\tilde{p}_z} - \hpropz
    z}>\dst\right\}\nonumber, 
\end{align*}
where $\II$ is the indicator function, and therefore, $\II\left\{\absv{\propp {\tilde{p}_z} - \hpropz
    z}>\dst\right\}=0$. This implies that 
$\absv{\propp {\tilde{p}_z} - \hpropz
    z}<\dst.$
By an identical reasoning, since $\dPsmbz>\delta/\beta$, we have 
$\absv{\propp {\dP} - \hpropz   z}<\dst$.
By the triangle inequality, 
\begin{align*}
\absv{\propp {\dP} -\propp {\tilde{p}_z}}\le \absv{\propp {\dP} - \hpropz
    z}+ \absv{\propp {\tilde{p}_z} - \hpropz z} <2\dst.
\end{align*}
Thus if $p(z) \geq \err/\beta$, then PML satisfies the required guarantee
with zero probability of error, 
and any error occurs only when  $\dP(z)<\err/\beta$. We bound this
probability as follows. When $Z\sim\dP$, 
\[
\probof{\dPsmbZ\le\err/\beta} \le \sum_{\smbz\in\cZ:\dPsmbz<\err/\beta}\dPsmbz \le
\err\cdot\absv{\cZ}/\beta. \eqed
\]